\documentclass[conference]{IEEEtran}
\makeatletter
\def\ps@headings{%
\def\@oddhead{\mbox{}\scriptsize\rightmark \hfil \thepage}%
\def\@evenhead{\scriptsize\thepage \hfil \leftmark\mbox{}}%
\def\@oddfoot{}%
\def\@evenfoot{}}
\makeatother

\usepackage{amsmath}
\usepackage{amsfonts}
\usepackage{amssymb}
\usepackage{amsthm}
\usepackage{epsfig}
\usepackage{latexsym}
\usepackage{graphicx}
\usepackage{comment}
\usepackage{float}
\usepackage[leftcaption]{sidecap}
\usepackage{xspace}

\newtheorem{theorem}{Theorem}
\newtheorem{lemma}[theorem]{Lemma}

\theoremstyle{definition}

\theoremstyle{remark}

\newcommand{\cm}[1]{}

\newcommand{\ve}{\varepsilon}

\newcommand{\rbar}{\bar{r}}

\begin{document}
\title{Energy-Delay Tradeoffs in a Load-Balanced Router}

\author{\IEEEauthorblockN{Matthew Andrews}
\IEEEauthorblockA{
Bell Labs, Murray Hill, NJ \\
andrews@research.bell-labs.com}
\and
\IEEEauthorblockN{Lisa Zhang}
\IEEEauthorblockA{
Bell Labs, Murray Hill, NJ \\
ylz@research.bell-labs.com}
\vspace{-1cm}
}
\date{}
\maketitle{}

\begin{abstract}
The Load-Balanced Router architecture has received a lot of attention 
because it does not require centralized scheduling at the internal switch
fabrics. In this paper we reexamine the architecture, motivated by its
potential to turn off multiple components
and thereby conserve energy in the presence of low traffic.

We perform a detailed analysis of the
queue and delay performance of a Load-Balanced Router under a simple random 
routing algorithm. We calculate probabilistic bounds for queue size
and delay, and show that the probabilities drop exponentially with
increasing queue size or delay.  We also demonstrate a tradeoff in
energy consumption against the queue and delay performance.  
\end{abstract}

\section{Introduction}

The concept of a {\em Load-Balanced Router} was
studied at length at the beginning of last decade. See for example
\cite{ChangLJ01,ChangLL02,KeslassyM02,KeslassyCYMHSM03,ChangLY03,ChangLS04,KeslassyCM04}.
In this work we analyze various performance aspects of a
Load-Balanced Router, motivated by the potential
energy saving enabled by this architecture. 

Energy efficiency in networking has recently attracted a large amount
of attention. One of the main aims in much of this work is captured by
the slogan {\em energy-follows-load}, also known as {\em
energy-proportionality}.  In other words, we wish to make sure that
the energy consumed by a networking device matches the amount of
traffic that the device needs to carry.  This is in contrast to more
traditional architectures for which the device operates at full rate
at all times even if it is lightly loaded.  Indeed, by a conservative
estimate in a study conducted by the Department of Energy in 2008, at
least 40\% of the total consumption by network elements such as
switches and routers can be saved if energy proportionality is
achieved. This translates to a saving of 24 billion kWh per year
attributed to data networking~\cite{Readout08}.  A recent
study~\cite{FranciniFKR11} further confirms that the power consumption of
some state-of-the-art commercial routers stays within a small
percentage of the peak power profile regardless of traffic
fluctuation, for example the significant daily variation in traffic
load \cite{RoughanGKRYZ02}.


Various approaches have been proposed in order to achieve
energy-proportionality.  {\em Speed scaling}, also known as {\em rate
adaptation}, and {\em powering down} are two popular methods for
effectively matching energy consumption to traffic load.  The former
refers to setting the processing speed of a network element according
to traffic load. It is typically assumed that the energy consumption is
superlinear with respect to the operating rate.  The latter refers to
turning off the element at certain times and so it either operates at
the full rate or zero rate. Both methods are the subject of active
research, though most of the work focuses on optimizing an individual
element in
isolation~\cite{IraniP05,Garrett08,YaoDS95,LiLY05,BansalKP07,
ChanCLLMW07,IraniSG07,IraniSG03,NedevschiPIRW08,FranciniS10}.  A central question to both methods is to
set the speed so as to minimize energy usage while maintaining a
desirable performance, e.g.\ latency or throughput.

\cm{
model is one
example. The basic idea is to turn off a device when it is not
actually processing traffic, which is most effective if most savings
comes from not powering the device.  In this model one of the key
questions is how often to change the status of a component since
typically there is a lag between when the component is powered up and
when it can actually start to process traffic.

A second approach that is often proposed for achieving
energy-proportionality is rate-adaptation. With this method the speed
of a component is adjusted to match the amount of traffic that needs
to be served.  It is typically assumed that the power usage of the
component is a superlinear function of its speed. For this method a
typical problem is to set the speed so as to minimize energy usage
while maintaining a desirable performance (e.g.\ latency or
throughput) of the component.

In this paper we consider another approach and address the energy
savings that can be achieved by examining the architecture of a
particular type of switch and handling the traffic in such a way that
enables us to turn off certain portions of the switch whenever the
traffic load is low. 
}

The Load-Balanced Router
architecture has the potential to
handle the traffic in such a way that portions of the device can be
turned off in response to lightly loaded traffic.  We provide what we
believe is the first detailed analysis of queue size and delay in a
Load-Balanced Router, and we provide a tradeoff between energy consumption and
queue size/delay.  In order to describe our results in more detail we
now give a brief description of the Load-Balanced Router architecture.

\subsection{Motivation for Traditional Load-Balanced Architecture}

One of the most fundamental goals of any router  architecture is to
achieve {\em stability}, which is sometimes referred to as 100\%
throughput. In other words the router aims to process all the
arriving traffic so long as no input and no output are inherently
overloaded.  The key difficulty with doing this is that the arriving
traffic may be highly non-uniform, i.e.\ if $A_{ik}$ is the
arrival rate for traffic going from input $i$ to output $k$, we will
typically have $A_{ik}\neq A_{i'k'}$ for $ik\neq
i'k'$. Early work on switching considered a crossbar architecture in
which matchings between the inputs and the outputs are set up at every
time step. 
It was shown in \cite{McKeownAW96} that the Maximum Weight Matching algorithm (with weights equal
to the backlog for each input-output pair) can ensure
stability. Subsequent papers looked at simplifications of this
scheduler that could still achieve stability.


However, a major drawback of all these approaches is that they require
a centralized scheduler with information about the backlogs of data on
each input-output pair. A solution is the Load-Balanced Router
that 
could make use of randomized
routing ideas first proposed by Valiant \cite{Valiant82}. In the
Load-Balanced Router there is a middle stage placed between the input
nodes and the output nodes. 
Each arriving
packet is routed to a middle-stage node chosen at random. After passing
through the middle stage each packet is then forwarded to its
designated output. In order to realize this architecture we place a
switching fabric between the input stage and middle stage and
between the middle stage and the output stage. The beauty of this
design is that the random routing ensures that for each of these
switching fabrics no complicated scheduling is needed. All we need to
do is repeat a uniform schedule in which each connection is served at
least once~\cite{ChangLJ01}.



\subsection{Energy Consideration for Revisiting Load-Balanced Architecture}

Our motivation for revisiting the Load-Balanced Router architecture is
that it provides an attractive framework for studying energy
proportionality~\cite{Antonakopoulos11}. 
For example, the number of active nodes in
the middle stage can be reduced in the presence of light traffic, and
increased with heavier traffic.  
The switch fabric between the input and middle stage and between the
middle stage and the output can be functionally viewed as full meshes,
as indicated in Figure \ref{f:lb2}.  
One possibility is to implement the mesh 
with round-robin crossbars. For example, Keslassy's thesis
\cite{KeslassyM02} 
assumes that the input, output and the middle stage are all
of size $n$ and each fabric is an $n\times n$ crossbar that operates in a time-slotted fashion. At each
time slot $t$ the first fabric connects input node $i$ to middle node
$(i+t)\mod n$ and the second fabric connects middle node $j$ to output
node $(j+t)\mod n$. In this implementation, if the number of active nodes
in the middle stage is reduced then the crossbar can either slowdown
or be turned off periodically.  

Adjusting the size of the middle stage or the speed of the switching
fabric requires considerable technological and engineering challenges.
%
In this note we do not aim to address these issues.
 Our focus is on analyzing queue size and delay given the
active portion of the middle stage, which leads to a tradeoff between
power consumption and the size of the active middle stage.


\subsection{Model and Definition} 
We formally define the Load-Balanced Router architecture as
follows. The router has $n$ inputs and $n$ outputs. We normalize the
line rate such that it equals 1 on each input and output link.  We
also have a middle stage lying in between the inputs and outputs that
consists of $m$ nodes. The traditional
Load-Balanced Router has $m=n$.  However, here we treat $m$ as a
separate parameter.  In between the input and the middle stage we have
an $n\times m$ mesh. Effectively each link in this mesh operates at rate
$\alpha /m$ for a {\em speedup} of $\alpha\ge 1$.  Similarly, in
between the middle and the output stage we have an $m\times n$ mesh
with link rate $\beta/m$ for a speedup factor of $\beta\ge 1$.  Each of the
two meshes is {\em input-buffered}, i.e.\ each input node has a
separate buffer for each middle node and each middle node has a
separate buffer for each output node. See Figure~\ref{f:lb2}.
\begin{figure}
\begin{center}
\includegraphics[width=3.in]{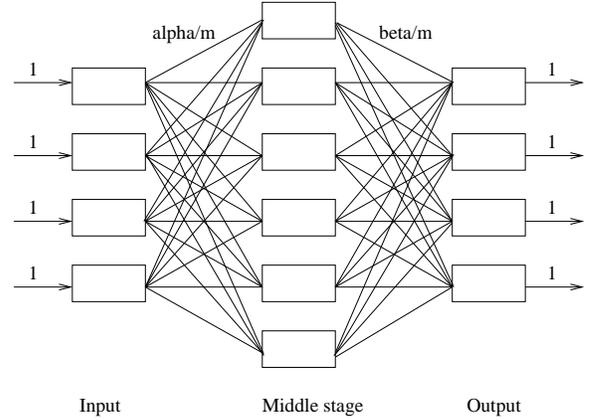}
\caption[]{A $4\times 6\times 4$ Load-Balanced Router 
architecture, where the number of nodes in the middle stage can
be different from the number of inputs and outputs.}
\label{f:lb2}
\end{center}
\end{figure}

From now on, we use $i$ to index the input, $j$ the middle stage and
$k$ the output.  We refer to the packets that wish to go from input
$i$ to output $k$ as $ik$ packets. Similarly, link $ij$ connects input
$i$ and node $j$ in the middle stage and link $jk$ connects node $j$
in the middle stage and output $k$. Buffer $ij$ (resp. $jk$) at node
$i$ (resp. $j$) buffers packets that are waiting to traverse link $ij$
(resp. $jk$).

The key to possible energy savings is that we assume that not all nodes need to be active at periods of low load. We suppose that 
at any time we can choose $m'\le m$ nodes to be active in the middle stage and that such a configuration requires power $w(m')$ for some function $w(\cdot)$. 

When an $ik$ packet $p$ arrives we choose a random middle node $j$
{\em among the active nodes in the middle stage} and place $p$ in
buffer at node $ij$.~\footnote{Note that round robin is another possibility
for choosing a middle-stage node. However, a random choice is more
robust against adversarial types of traffic arrivals.  We do not go 
into details here.} The link $ij$ operates continuously at rate
$\alpha/m$ and transmits packets in the $ij$ buffer in a FIFO manner.
Similarly, when $p$ arrives at node $j$ in the middle
stage, it is placed in the buffer $jk$. The link $jk$ continuously
operates at rate $\beta/m$ and transmits packets in the $jk$ buffer in
a FIFO manner.  As we can see the scheduling for both stages requires
no centralized intelligence and is extremely simple.

Lastly we describe our traffic model.  As is common in work on scheduling in routers, we assume that packets are of unit size (or else are partitioned into cells of unit size).  
For any $s,t$ let $A_{ik}(s,t)$ be the
amount of $ik$ traffic arriving at the router in the time interval
$[s,t)$ and let $A_i(s,t)=\sum_k A_{ik}(s,t)$ and $A_k(s,t)=\sum_i
A_{ik}(s,t)$.  We assume that the $ik$ traffic, the input $i$ traffic
and the output $k$ traffic are $(\sigma_{ik},r_{ik})$, $(\sigma,1)$ and
$(\sigma,1-\ve)$ constrained respectively for some burst parameters
$\sigma_{ik}$ and $\sigma$, for some rate parameters $r_{ik}$ and
for some load parameter $\ve$.  In other words we assume that,
\begin{eqnarray*}
A_{ik}(s,t)&\le& \sigma_{ik}+r_{ik}(t-s)\\
A_{i}(s,t)&\le& \sigma+(t-s)\\
A_{k}(s,t)&\le& \sigma+(1-\ve)(t-s).
\end{eqnarray*}

We remark that the arrival rates $r_{ik}$ will typically vary over
time. Indeed, the energy savings that we hope to gain come precisely
from the fact that we can match the number of active components to the
traffic.  However, we assume that this happens over a slow timescale
and so we perform our scheduling analysis as if the rates $r_{ik}$ are
fixed.

\subsection{Results}

\begin{itemize}
\item 
In Section~\ref{s:rate-only} we make the simple statement that 
$\lceil \rbar m\rceil$ active middle-stage nodes suffice
for handling the traffic load
where $\rbar\ge \sum_i r_{ik}$ for all $k$ and $\rbar
\ge \sum_k r_{ik}$ for all $i$. This in turn  implies that the energy required to serve the traffic over the long-term is 
$w(\lceil \rbar m\rceil)$. 
\cm{
We also discuss the typical energy savings that would
come from turning off the remaining middle stages. These energy
savings would come from two sources.  First of all we would not need
to operate the middle stages themselves (and all they buffering that
they employ).  Second, if we have a smaller number of middle stages
then we can run the two meshes at a slower rate (or we could
periodically turn the meshes off for a number of timeslots).  (We
remark that changing the speed of the meshes would require some
centralized knowledge of the maximum loaded input or output port.)
However, since we assume that the $ik$ arrival rates $r_{ik}$ are
fixed over the time constants that we are considering, it would
require a minimum of signaling to convey this load information.
}
\item 
In Sections~\ref{s:overview} to \ref{s:delay-exact} we present a probabilistic analysis that bounds the queue sizes at the input and middle stages and bounds the delay experienced by packets as they travel from the 
router input to the router output. We first 
bound the probability for a queue size to exceed a certain amount
$q$, and the  delay to exceed a certain amount of $d$, assuming a
fixed sized number of middle stages.  An important feature of our bounds is that they decrease
exponentially with $q$ and $d$.  For a fixed traffic load, we then
derive a trade off between the queue size/delay performance and the
number of active middle-stage nodes which in turn gives a
energy-delay and energy-queue tradeoff.


Leaving energy minimization aside, we believe that this is the 
first detailed analysis of the delay and queue performance of a
Load-Balanced Router, which may be interesting in its own right.

\item 
In Section~\ref{s:numerical} we present some numerical examples
to validate our analytical findings. 
\end{itemize}

We note that our approach is different from the traditional powerdown
and rate adaptation techniques since we will be directing traffic in
such a way that enables some components to be off. In
other words for the middle stage nodes we are not trying to match service
rate to a traffic process that is exogenous. We are trying to match
the active middle stage nodes to a traffic process that is under our control due to our
ability to route within the router.

We also remark that our bounds could be used to govern how many middle
nodes are active in a Load-Balanced Router without necessarily
computing all the bounds on the fly. We could instead precompute the bounds and create a simple
look-up table that determines how many middle stage nodes should be active based on measurements of the load at the inputs and outputs. 




\section{Throughput Analysis}
\label{s:rate-only}
Recall that $r_{ik}$ represents the current arrival rate of traffic
that wishes to be routed from input $i$ to output $k$. Let $r_i =
\sum_k r_{ik}$ and let $r_k=\sum_i r_{ik}$. 
Let $\rbar$ be such that $r_i\le \rbar$ and $r_k\le \rbar$ for all
$i,k$. 
\begin{lemma}
If we use $m'$ middle stage elements then the router is not overloaded
if $m'\ge \rbar m$. Hence the power required to serve all the traffic in the long-term is at most $w(\lceil \rbar m\rceil)$. 
\end{lemma}
\begin{proof}
Follows from the fact that if we turn on $m'$ middle stage elements
then for each middle element $j$, $1\le j \le m'$, the traffic rate
that will be routed on link $ij$ will be at most $\rbar/m'$ which
by assumption is at most $(m'/m)/m'=1/m$. Since the capacity on the link
$ij$ is $\alpha/m$ for some $\alpha>1$, this implies that the $ij$
link is not overloaded. A similar argument applies to each 
link $jk$ between the middle and output stages
\end{proof}

\section{Overview of Delay Analysis}

\label{s:overview}
Before diving into details of the delay analysis,
we first provide a high level overview of
our techniques.

\cm{
The analysis of Section~\ref{s:rate-only} in which we only
focus on how many middle stages are needed to prevent the switch
becoming overloaded is very simple and essentially just considers the
aggregate rate through each of the switch elements. We therefore focus
on the techniques required to perform the analysis of
Sections~\ref{s:delay-approx} and \ref{s:delay-exact} in which we wish
to trade-off the number of active middle stages with the delay
experienced by data as it passes through the switch. 
}

\subsection{Relationship with Stochastic Network Calculus}

\cm{
Our formulas in Section~\ref{} that look at the energy required to
prevent the switch from being overloaded come from simple calculations
based on the arrival rates $r_{ik}$. The meat of our analysis is
therefore focused on the tradeoff between delay and the number of
middle stages (and hence energy). To do this we}

We use a variant of network calculus~\cite{Cruz91a,Cruz91b,LeBoudec} sometimes referred to as {\em
stochastic network calculus}.  The original form of
network calculus derives delay bounds by imposing {\em upper bounds}
on the amount of traffic arriving at a node via {\em arrival curves}
and {\em lower bounds} on the amount of traffic served by a node {\em
via service curves}.  By relating these two curves we can both obtain
a bound on the delay suffered by data at a network element and also
characterize the arrival curves for the data at any downstream nodes.
However, in the traditional network calculus all such bounds are
required to hold with probability 1.  In our context this will lead to
extremely weak bounds since there is a non-zero probability that the
router will send a large number of packets to a single middle-stage node,
thus condemning them all to extremely poor service.

\cm{
An alternative therefore is to use a {\em stochastic network calculus}
in which we only wish for bounds on service to hold with high
probability.  A detailed formulation of a stochastic network calculus
was outlined in a series of papers by Jiang and
others~\cite{Jiang05analysisof,jiang06}. However, we were unable to
use these techniques directly to obtain tight bounds.  At a high level
the reason is that Jiang's approach tries to obtain curves that bound
the probability that the delay suffered by a flow at a network element
(the first stage of the load-balanced router) is more than a certain
amount $d$.  They then use these curves to bound the worst-case
arrivals at downstream elements (in our case this could be the middle
stage of the router.) However, in our case we obtain better bounds by
not utilizing the earlier delay curves to bound the behavior at
downstream nodes and simply basing all our calculations on an analysis
of the external arrival patterns of various ensembles of flows.  These
ensembles can be analyzed directly via Chernoff bounds.  (If we had to
work directly with the probabilistic arrival curves then we would have
had to handle more complicated convolutions arrival and service
curves. {\sc do we really want to say that?}  We elaborate further on
this distinction in the initial parts of Section~\ref{s:delay-approx},
where we give an overview of the types of bounds we use.)
}

An alternative therefore is to use a {\em stochastic network calculus}
in which we only wish for bounds on service to hold with high
probability.  A detailed formulation of a stochastic network calculus
was outlined in a series of papers by Jiang and
others~\cite{Jiang05analysisof,jiang06}.  At a high level Jiang's
approach obtains curves that bound the probability that the delay
exceeds a certain amount at upstream elements, and then use these
curves to bound the worst-case arrivals at downstream
elements. However, we follow a slightly different approach since we
are able to obtain better bounds by not directly utilizing the service
curves at the input nodes to bound the arrivals at the middle-stage
nodes.  We instead base our calculations at the middle stage on the
external arrivals of various ensembles of flows that might then be
time-shifted due to delays at the input.  This allows us to avoid handling 
complicated convolutions of arrival and service curves.
We elaborate further on this distinction later.

\subsection{Our Approach}
We divide our analysis into a series of pieces. 
\paragraph{Bound the queue build-up at the input} 
Recall that between input $i$ and middle-stage node $j$ we effectively have
a link with speed $\alpha/m$. Also recall that the input has a buffer
especially dedicated to the traffic that wishes to go from input $i$
to middle-stage node $j$. Suppose that at some time $t$ this buffer has
level $q$.  Let $s$ be the last time that this buffer was empty. Note
that during the time interval $[s,t)$ the $ij$ link served data of
total size $\alpha(t-s)/m$. Therefore the total data that arrived for
link $i$ during the time interval $[s,t)$ is at least
$q+(\alpha(t-s)/m)$.

Therefore the probability that link $i$ has a backlog of size $q$ at
time $t$ is upper bounded by the probability that {\em for some $s\le
t$}, the amount of $ik$ data arriving at link $i$ during the interval
$[s,t)$ is at least $q+(\alpha(t-s)/m)$. However, recall that the
traffic arriving to input $i$ arrives at rate at most $\rbar$ 
and each packet is sent to a
middle-stage node chosen uniformly at random. Hence, for fixed $s$
and $t$, we can use a Chernoff bound to bound the probability that the
amount of $ik$ data arriving during the interval $[s,t)$ is at least
$q+(\alpha(t-s)/m)$. It is easy to show that this probability
decreases exponentially in $s$ and so we can use a union bound to
bound the probability that this occurs for {\em any} $s\le t$. 

\paragraph{Bound the delay experienced at the input} 
Translating our bound on queue size into a bound on delay is
simple. Since the transmission rate on the $ij$ link is $\alpha/m$,
the event that the head of line packet for the $ij$ link at time $t$
has experienced delay $d$, implies the event that at time $t-d$, the
queue size for the $ij$ link was at least $d\alpha/m$. An upper bound
on the probability of the latter event can be derived using the method
described earlier.

\paragraph{Bound the queue build-up at middle stage} 
We now give an overview of how we perform
the analysis at the head of the $jk$ queue at middle element $j$. This
forms the crux of our analysis since we are in a more complicated
situation due to the fact that the packet arrivals at the middle stage
are affected by how they are served at the input. 
As before note that if the queue for link $jk$ has size $q$ at time
$t$, then for some time $s\le t$, the arrivals for link $jk$ at middle-stage 
node $j$ during the time interval $[s,t)$ must be at least
$q+(\beta(t-s)/m)$. Now suppose that the oldest of this data arrived
at the input at time $s-d$, and suppose in addition that this
data arrived at middle-stage node $j$ on link $ij$.  We can therefore state
that the total amount of data arriving at the system in the time
interval $[s-d,t)$ that is destined for link $jk$ is at least
$q+(\beta(t-s)/m)$ and the delay experienced by data arriving at link
$ik$ at time $s-d$ is at least $d$.  Via union bounds we can calculate
the probability that this occurs for {\em any} $i,s,d$. In particular,
for $d$ small we can say that the probability that the arrivals exceed
$q+(\beta(t-s)/m)$ is small whereas if $d$ is large then the
probability that the link $ik$ delay is at least $d$ is small.

We make three points about this analysis at the middle stage.
\begin{itemize}

\item 
This is where our analysis deviates slightly from the traditional
methodology of network calculus. We do not calculate a service curve
for the input and use that service curve to bound the arrivals
at the middle stage. 
Instead we analyze the delay behavior at the
input and then use that calculation to bound the middle stage
queue size using an expression that {\em still involves arrivals at
the input}.

\item 
Our initial analysis makes a slight approximating assumption in that 
for the $ijk$ defined above, we treat the
arriving traffic for link $ij$ and the delay behavior on link $ij$ as
independent. In reality of course they will be correlated due to
traffic that passes from input $i$ to output $k$ through middle-stage node
$j$.  However, this approximation will typically not have a large
effect for larger values of $m$ since the $ijk$ traffic only forms a
$1/m$ fraction of the total $kj$. However, in order to get a more
accurate bound, in Section~\ref{s:delay-exact} we present a more
detailed expression that deals with the correlation explicitly.

\item 
Our analysis is based on Chernoff bounds. However, the form 
of the Chernoff bound that we use changes depending on
whether the bound on data arrivals that we are considering for a
particular link is close to or far from the expected number of data
arrivals for that link.  This unfortunately leads to a somewhat
involved case analysis.
\end{itemize}

\paragraph{Bound the delay experienced at the middle stage} 
The conversion of the queue-size bound to a delay bound at the middle stage can be done in
the exact same way as for the input. Now that we have an expression for the delay at both stages we can convert it into an 
expression for the end-to-end delay from the inputs to the outputs.  


\section{Analytical Bounds on Delay}
\label{s:delay-approx}
We now present the details of our delay analysis. 
We rely heavily on the following Chernoff bounds~\cite{Scheideler00}. In particular we use (\ref{eq:cb2}) to derive
analytical bounds and we use (\ref{eq:cb1}) to derive tighter bounds for 
numerical simulation. 
\begin{theorem}[Chernoff Bound]
Let $X_1, \dots, X_n$ be independent binary random variables, and let
$\mu$ be an upper bound on the expectation $E[\sum_i X_i]$.  For all $\delta > 0$,
\begin{eqnarray}
Pr[\sum_{i} X_i \ge (1+\delta)\mu] & \le & 
\left({{e^\delta} \over {(1+\delta)^{1+\delta}}} \right)^\mu \label{eq:cb1}\\
&\le& e^{-\min(\delta^2,\delta)\cdot \mu/3}\label{eq:cb2}
\end{eqnarray}
\end{theorem}
In what follows we sometimes refer to $\mu$ as the {\em aggregated
mean} and $\delta$ as the {\em excess factor}. For conciseness we
shall often use the aggregated mean even though $\mu$ is strictly
speaking a bound on the mean.

\subsection{Input stage analysis}
We begin by computing the queue distribution at the head of link $ij$,
where $i\in [1,n]$ is an input and $j \in [1,m]$ is in the middle stage.  
As in Section~\ref{s:rate-only} we assume that $r_i\le \rbar$ and
$r_k\le \rbar$ for all $i,k$ and we assume that $m'$ middle stage
elements are currently active.  Initially, in order to keep the
formulas manageable, we shall derive our formulas for the case in
which $\rbar=1$ and $m'=m$. We shall also assume that the arriving
traffic is smooth and so we do not have the burst terms $\sigma_i$ and
$\sigma_k$. Later on, we shall show how to adapt the formulas when
these assumptions do not hold.

Let $A_{ijk}(t)$ be the binary random variable that indicates whether
a packet with input $i$ output $k$ is mapped to middle stage $j$ at
time $t$.  We use $A_{ijk}(t_1,t_2)$ to denote $\sum_{t=t_1}^{t_2}
A_{ijk}(t)$, the total arrival in the duration $(t_1,t_2]$.  Let
$Q^{(1)}_{i,j}(t)$ be the random variable for the queue size at the
head of link $ij$ at time $t$. To compute $Q^{(1)}_{i,j}(t)$, let us
assume $s<t$ be the last time that the queue at $ij$ is empty.  For
$Q^{(1)}_{i,j}(t)$ to be larger than a value $q$, the arrival
$A_{ijk}(s,t)$ over all $k\in [1,n]$ must be at least $\alpha\cdot
{{t-s}\over m}+q$ since link $ij$ is operated at a rate $\alpha\over
m$. Formally, 
\begin{eqnarray}
Pr[Q^{(1)}_{ij}(t)\ge q] \le 
\sum_{s\le t} Pr \left[\sum_{k=1}^n A_{ijk}(s,t) 
\ge \alpha\cdot {{t-s}\over m}+q\right]
\label{eq:q1}
\end{eqnarray} 
We now bound the right-hand-side of (\ref{eq:q1}). Since $A_{ijk}(t)$
are independent binary variables, we apply the Chernoff bound to every
term in the summation.  The following is the aggregated mean $\mu$
and the excess factor $\delta$ for the above probability. 
\begin{eqnarray*}
\left\{
\begin{array}{lll}
\mu & = & {{t-s}\over m} \\
\delta & = & \alpha -1 +{{qm}\over{t-s}}
\end{array}
\right.,
\end{eqnarray*}
since $\alpha\cdot {{t-s}\over m}+q = \mu(1+\delta)$. 
If $\alpha\ge 2$, we have $\delta\ge 1$ and we use~(\ref{eq:cb2}) to derive,
\begin{eqnarray}
Pr[Q^{(1)}_{ij}(t)\ge q] 
\le \sum_{t-s=0}^{\infty} e^{-{{t-s}\over{3m}}(\alpha-1)-{q\over 3}}
\le {{e^{-{q\over 3}}}\over {1-e^{-{{\alpha-1}\over{3m}}}}}. \label{eq:lessthan2}
\end{eqnarray}
If $\alpha <2$, we have $\delta \ge 1$ when $t-s\le {{qm}\over{2-\alpha}}$.
Otherwise $\delta<1$.  We apply~(\ref{eq:cb2}) as follows.
\begin{eqnarray}
&& Pr[Q^{(1)}_{ij}(t)\ge q]\nonumber \\
& \le & \sum_{t-s=0}^{{qm}\over{2-\alpha}} 
	e^{-{{t-s}\over{3m}}(\alpha-1)-{q\over 3}} +
\sum_{t-s={{qm}\over{2-\alpha}}}^{\infty}
	e^{-{{t-s}\over{3m}}(\alpha-1)^2} \nonumber\\
& \le &
{{e^{-{q\over 3}}}\over {1-e^{-{{\alpha-1}\over{3m}}}}}
+{{e^{-{{(\alpha-1)^2}\over{3(2-\alpha)}}q}
\over{1-e^{-{(\alpha-1)^2}\over{3m}}}}}\label{eq:morethan2}
\end{eqnarray}

Note that both expressions are exponentially decreasing in $q$. For the time being we shall proceed according to the case $\alpha\ge
2$ since this leads to more manageable formulas. Later on, we shall
indicate where to adapt the formulas when we are in a scenario where
$\alpha<2$. 

Let $D^{(1)}_{ij}(t)$ be the maximum delay that some packet has
experienced at time $t$ in the queue $Q^{(1)}_{ij}(t)$.  For this
delay to be more than $d$ at time $t$, some packet must be in the
queue at time $t-d$. Since link $ij$ operates at rate ${\alpha\over
m}$ in a FIFO manner, the queue at time $t-d$ must be at least
$d\alpha/m$. Therefore,
\begin{eqnarray*}
Pr[D^{(1)}_{ij}(t)\ge d] \le Pr[Q^{(1)}_{ij}(t-d)\ge d\alpha/m] 
\le {{e^{-{d\alpha}\over {3m}}}\over {1-e^{-{{\alpha-1}\over{3m}}}}}.
\end{eqnarray*}
By a union bound,
\begin{eqnarray}
Pr[\exists i,~D^{(1)}_{ij}(t)\ge d] 
\le n\cdot {{e^{-{d\alpha}\over {3m}}}\over {1-e^{-{{\alpha-1}\over{3m}}}}}.
\label{eq:d1}
\end{eqnarray}

Recall that the above analysis was performed in the absence of the
burst term $\sigma$.  If we do have bursty traffic then we can adjust
the formulas by making the following changes to the aggregated mean
and the excess factor and propagating these changes through the resulting formulas. 
\begin{eqnarray*}
\left\{
\begin{array}{lll}
\mu & = & {{t-s+\sigma}\over m} \\
\delta & = & \frac{(t-s)(\alpha-1)+qm-\sigma}{t-s+\sigma}
\end{array}
\right.
\end{eqnarray*}

\subsection{Middle stage analysis}

We now compute the queue distribution at the head of link $jk$, where
$j\in [1,m]$ is in the middle stage and $k\in [1,n]$ is an output.
Let $Q_{jk}^{(2)}(t)$ be defined similarly as $Q_{ij}^{(1)}(t)$.  To
bound $Q_{jk}^{(2)}$ at time $t$, let $s\le t$ be the last time that
the queue $Q_{jk}^{(2)}$ was empty. For $Q_{jk}^{(2)}(t)$ to be larger
than $q$, there must be at least $\frac{(t-s)\beta}{m}+q$ distinct
packets in $Q_{jk}^{(2)}$ during the time period $[s,t]$. Further, let $s-d$
be the earliest time one of these packets arrived at an input, say
$i$. This packet must experience a delay of at least $d$ in
$Q_{ij}^{(1)}$.
Therefore,
\begin{eqnarray*}
&& Pr[Q_{jk}^{(2)}(t)\ge q] \\
&\le&
\sum_d\sum_{s\le t} Pr\left[\sum_i A_{ijk}(s-d,t)\ge 
\frac{(t-s)\beta}{m}+q\right] \\
&& \hspace{1.2in}
\cdot Pr[\exists i,~D_{ij}^{(1)}(s)\ge d] \\
& \le &
\sum_d {{ne^{-{d\alpha}\over {3m}}}\over {1-e^{-{{\alpha-1}\over{3m}}}}}\\
&& \hspace{.2in}
\cdot \sum_{s\le t} Pr\left[\sum_i A_{ijk}(s-d,t)\ge 
\frac{(t-s)\beta}{m}+q\right]
\end{eqnarray*}
Note that the bound~(\ref{eq:d1}) on the delay distribution at the input stage is independent
of the time index. We can therefore move $Pr[\exists
i,~D_{ij}^{(1)}((s))\ge d]$ to outside the summation indexed by time
in the second inequality above.

We proceed to bound $\sum_{s\le t} Pr\left[\sum_i A_{ijk}(s-d,t)\right.$ $\ge$
$\left.\frac{(t-s)\beta}{m}+q\right]$ based on the following expressions for the 
expectation $\mu$ and the excess factor $\delta$.
\begin{eqnarray*}
\left\{
\begin{array}{lll}
\mu & = & {{t-s+d}\over m} \\
\delta & = & \frac{(t-s)(\beta-1)+qm-d}{t-s+d}
\end{array}
\right.,
\end{eqnarray*} 
since $(1+\delta)\mu = \frac{(t-s)\beta}{m}+q$.
There are two cases to consider,
$\beta \le 2$ or $\beta > 2$.  

For $\beta\le 2$, we further consider the
following subcases, depending on ${qm\over d}-1$, the value
of $\delta$ when $t-s=0$. Note that as $t-s$ increases, the value of
$\delta$ approaches $\beta-1$. 

\begin{itemize}
\item {\em Case 1a: $1 \le \frac{qm}{d}-1$, and $t-s\le
\frac{qm-2d}{2-\beta}$}. In this case $1\le \delta$. Since $\delta \mu
= \frac{(t-s)(\beta-1)+qm-d}{m}$, bound~(\ref{eq:cb2}) implies (\ref{eq:1a}).

\item {\em Case 1b: $1\le \frac{qm}{d}-1$, and $\frac{qm-2d}{2-\beta}
\le t-s$.}  In this case $\beta-1\le \delta \le 1$,
which implies $\delta^2\mu \ge (\beta-1)^2\mu$.
Bound~(\ref{eq:cb2}) in turn implies (\ref{eq:1b}).

\item {\em Case 2: $\beta-1\le \frac{qm}{d}-1\le 1$, and for all values 
of $t-s$.}  In this case,
$\beta-1\le \delta \le 1$, which is the same situation as 1b
and implies (\ref{eq:2}) in the same way.

\item {\em Case 3a: $\frac{qm}{d}-1\le \beta-1$, 
$\frac{d(\beta+1)-2qm}{\beta-1}\le t-s$
and $\frac{d(\beta+1)-2qm}{\beta-1} \le 0$}.  In this case
$\frac{\beta-1}{2}\le \delta\le \beta-1$ for all values of $t-s$.
Since $\delta^2\mu \ge {(\beta-1)^2\mu\over 4}$, bound~(\ref{eq:cb2})
implies (\ref{eq:3a}).

\item {\em Case 3b: $\frac{qm}{d}-1\le \beta-1$, 
and $0<\frac{d(\beta+1)-2qm}{\beta-1}\le t-s$}.  In this case
$\frac{\beta-1}{2}\le \delta\le \beta-1$
for $t-s\ge \frac{d(\beta+1)-2qm}{\beta-1}$.
Since $\delta^2\mu \ge {(\beta-1)^2 \mu\over 4}$, bound~(\ref{eq:cb2})
implies (\ref{eq:3b}).

\item{\em Case 3c: $\frac{qm}{d}-1\le \beta-1$, and $t-s<
\frac{d(\beta+1)-2qm}{\beta-1}$}. We trivially upper bound the
probability by 1 as in~(\ref{eq:3c}).
\end{itemize}

\begin{eqnarray}
&& Pr[Q_{jk}^{(2)}(t)\ge q] \nonumber \\
&\le&
\sum_d\sum_{s\le t} Pr[\sum_i A_{ijk}^{(1)}(s-d,t)\ge \frac{(t-s)\beta}{m}+q]\\
&& \hspace{1.3in}\cdot Pr[\exists i,~D_{ij}^{(1)}((s))\ge d]
\nonumber 
\\
\mbox{\em{1a}}
&\le&
\sum_{d=0}^{\frac{qm}{2}}
\left(\frac{ne^{-\frac{d\alpha}{3m}}}{1-e^{-\frac{\alpha-1}{3m}}}\right) 
\sum_{t-s=0}^{\frac{qm-2d}{2-\beta}}
e^{-\frac{1}{3}\frac{(t-s)(\beta-1)+qm-d}{m}} +
\label{eq:1a}
\\
\mbox{\em{1b}}
&&\sum_{d=0}^{\frac{qm}{2}}
\left(\frac{ne^{-\frac{d\alpha}{3m}}}{1-e^{-\frac{\alpha-1}{3m}}}\right)
\sum_{t-s=\frac{qm-2d}{2-\beta}}^{\infty}
e^{-\frac{1}{3}(\beta-1)^2\frac{t-s+d}{m}} +
\label{eq:1b}
\\
\mbox{\em{2}}
&&\sum_{d=\frac{qm}{2}}^{\frac{qm}{\beta}}
\left(\frac{ne^{-\frac{d\alpha}{3m}}}{1-e^{-\frac{\alpha-1}{3m}}}\right)
\sum_{t-s=0}^{\infty}
e^{-\frac{1}{3}(\beta-1)^2\frac{t-s+d}{m}}+
\label{eq:2}
\\
\mbox{\em{3a}}
&&\sum_{d=\frac{qm}{\beta}}^{\frac{2qm}{\beta+1}}
\left(\frac{ne^{-\frac{d\alpha}{3m}}}{1-e^{-\frac{\alpha-1}{3m}}}\right)
\sum_{t-s=0}^{\infty}
e^{-\frac{1}{3}\frac{(\beta-1)^2}{4}\frac{t-s+d}{m}}+
\label{eq:3a}
\\
\mbox{\em{3b}}
&&\sum_{d=\frac{2qm}{\beta+1}}^{\infty}
\left(\frac{ne^{-\frac{d\alpha}{3m}}}{1-e^{-\frac{\alpha-1}{3m}}}\right)
\nonumber \\
&& \hspace{1cm}\sum_{t-s=\frac{d(\beta+1)-2qm}{\beta-1}}^{\infty}
e^{-\frac{1}{3}\frac{(\beta-1)^2}{4}\frac{t-s+d}{m}}+
\label{eq:3b}
\\
\mbox{\em{3c}}
&&\sum_{d=\frac{2qm}{\beta+1}}^{\infty}
\left(\frac{ne^{-\frac{d\alpha}{3m}}}{1-e^{-\frac{\alpha-1}{3m}}}\right)
\sum_{t-s=0}^{\frac{d(\beta+1)-2qm}{\beta-1}} 1
\label{eq:3c}
\end{eqnarray}

\cm{
\\
\mbox{\em{1a}}
&\le&
\sum_{d=0}^{\infty}
\left(\frac{ne^{-\frac{d\alpha}{3m}}}{1-e^{-\frac{\alpha-1}{3m}}}\right)
\left(\frac{e^{-\frac{1}{3}(q-\frac{d}{m})}}{1-e^{-\frac{\beta-1}{3m}}}\right)+
\nonumber
\\
\mbox{\em{1b}}
&&
\sum_{d=0}^{\infty}
\left(\frac{ne^{-\frac{d\alpha}{3m}}}{1-e^{-\frac{\alpha-1}{3m}}}\right)
\left(\frac{e^{-\frac{1}{3}(\beta-1)^2(q-\frac{d}{m})}}{1-e^{-\frac{(\beta-1)^2}{3m}}}\right)+
\nonumber
\end{eqnarray}
\begin{eqnarray*}
\mbox{\em{2}}
&&
\sum_{d={{qm}\over 2}}^{\infty}
\left(\frac{ne^{-\frac{d\alpha}{3m}}}{1-e^{-\frac{\alpha-1}{3m}}}\right)
\left(\frac{e^{-\frac{1}{3}(\beta-1)^2\frac{d}{m}}}{1-e^{-\frac{(\beta-1)^2}{3m}}}\right)+
\\
\mbox{\em{3a}}
&&
\sum_{d=\frac{qm}{\beta}}^{\infty}
\left(\frac{ne^{-\frac{d\alpha}{3m}}}{1-e^{-\frac{\alpha-1}{3m}}}\right)
\left(\frac{e^{-\frac{1}{12}(\beta-1)^2\frac{d}{m}}}{1-e^{-\frac{(\beta-1)^2}{12m}}}
\right)+
\\
\mbox{\em{3b}}
&&
\sum_{d=\frac{2qm}{\beta+1}}^{\infty}
\left(\frac{ne^{-\frac{d\alpha}{3m}}}{1-e^{-\frac{\alpha-1}{3m}}}\right)
\left(\frac{e^{-\frac{1}{6}(\beta-1)\frac{d\beta-qm}{m}}}{1-e^{-\frac{(\beta-1)^2}{12m}}}
\right)+
\\
\mbox{\em{3c}}
&&
\sum_{d=\frac{2qm}{\beta+1}}^{\infty}
\left(\frac{ne^{-\frac{d\alpha}{3m}}}{1-e^{-\frac{\alpha-1}{3m}}}\right)
\left(\frac{d(\beta+1)-2qm}{\beta-1}\right)
\\
}
\begin{eqnarray*}
\mbox{\em{1a}}
&\le&
\left(\frac{n}{1-e^{-\frac{\alpha-1}{3m}}}\right)
\left(\frac{e^{-\frac{q}{3}}}{1-e^{-\frac{\beta-1}{3m}}}\right)
\left(\frac{1}{1-e^{-\frac{\alpha-1}{3m}}}\right) + 
\\
\mbox{\em{1b}}
&&
\left(\frac{n}{1-e^{-\frac{\alpha-1}{3m}}}\right)
\left(\frac{e^{-\frac{1}{3}(\beta-1)^2q}}{1-e^{-\frac{(\beta-1)^2}{3m}}}\right)
\left(\frac{1}{1-e^{-\frac{\alpha-(\beta-1)^2}{3m}}}\right) +
\\
\mbox{\em{2}}
&&
\left(\frac{ne^{-\frac{\alpha}{6}q}}{1-e^{-\frac{\alpha-1}{3m}}}\right)
\left(\frac{e^{-\frac{1}{3}(\beta-1)^2q}}{1-e^{-\frac{(\beta-1)^2}{3m}}}\right)\left(\frac{1}{1-e^{-\frac{\alpha-(\beta-1)^2}{3m}}}\right)
+
\\
\mbox{\em{3a}}
&&
\left(\frac{ne^{-\frac{\alpha}{3\beta}q}}{1-e^{-\frac{\alpha-1}{3m}}}\right)
\left(\frac{e^{-\frac{1}{12\beta}(\beta-1)^2q}}{1-e^{-\frac{(\beta-1)^2}{12m}}}
\right)
\left(\frac{1}{1-e^{-\frac{4\alpha+(\beta-1)^2}{12m}}}\right)+
\\
\mbox{\em{3b}}
&&
\left(\frac{ne^{-\frac{2\alpha}{3(\beta+1)}q}}{1-e^{-\frac{\alpha-1}{3m}}}\right)
\left(\frac{e^{-\frac{1}{6}(\beta-1)^2q}}{1-e^{-\frac{(\beta-1)^2}{12m}}}
\right)
\left(\frac{1}{1-e^{-\frac{2\alpha+\beta(\beta-1)}{6m}}}\right)+
\\
\mbox{\em{3c}}
&&
\left(\frac{(\beta+1)ne^{-\frac{\alpha}{3m}}}{(\beta-1)(1-e^{-\frac{\alpha-1}{3m}})}\right) \left(\frac{1}{(1-e^{-\frac{\alpha}{3m}})^2}\right)\\
&& \left(\left(\frac{2qm}{\beta+1}\right)e^{-\frac{\alpha}{3m}(\frac{2qm}{\beta+1}-1)}-
\left(\frac{2qm}{\beta+1}-1\right)e^{-\frac{\alpha}{3m}(\frac{2qm}{\beta+1})}\right)
\end{eqnarray*}
Note that every term above decreases exponentially with the queue size
$q$.

The case in which $\beta>2$ is simpler. We omit the analysis for
space consideration.
Recall that all of the above formulas are for the case $\alpha\ge 2$. If $\alpha<2$ then we need to replace all the factors (\ref{eq:lessthan2}),
$$
\frac{e^{\frac{-d\alpha}{3m}}}{1-e^{-\frac{\alpha-1}{3m}}}
$$
with (\ref{eq:morethan2})
$$
\frac{e^{\frac{-d\alpha}{3m}}}{1-e^{-\frac{\alpha-1}{3m}}}+
\frac{e^{\frac{-d\alpha(\alpha-1)^2}{3m(2-\alpha)}}}{1-e^{-\frac{(\alpha-1)^2}{3m}}}.
$$
If output $k$ has a burst term $\sigma$ then as in the input stage we can reflect this by adjusting the aggregated mean and the excess factor to,
\begin{eqnarray*}
\left\{
\begin{array}{lll}
\mu & = & {{t-s+d+\sigma}\over m} \\
\delta & = & \frac{(t-s)(\beta-1)+qm-d-\sigma}{t-s+d+\sigma}
\end{array}
\right..
\end{eqnarray*} 

We conclude this section by bounding the delay distribution at the middle stage.
$D_{jk}^{(2)}$ be the maximum delay that some packet has experienced
at time $t$ in the queue $Q_{jk}^{(2)}(t)$.  For this delay to be more
than $d$ at time $t$, some packet must be in the queue at time
$t-d$. Since link $jk$ operates in a FIFO manner, the queue at time
$t-d$ must be at least $d\beta/m$.
Hence $Pr[D_{jk}^{(2)}(t)\ge d]\le Pr[Q_{jk}^{(2)}(t-d)\ge d\beta/m]$.


We stress again that all of the above formulas are for the case $\alpha\ge 2$. If $\alpha<2$ then we need 
to replace all the factors (\ref{eq:lessthan2}) with factors (\ref{eq:morethan2}).

\subsection{Eventual end-to-end delay}
Now that we have delay bounds for the two stages of the router 
we can obtain a bound on the end-to-end delay distribution. 
In the following we let $g_{m,\beta}(q)$ be a shorthand for the upper bound
that we have derived on $Pr[Q_{jk}^{(2)}(t)\ge q]$ and $f_{m,\alpha}(q)$ a
shorthand for our upper bound on $Pr[Q_{ij}^{(1)}(t)\ge q]$. 
Suppose that an $ijk$ packet is still traversing the router at time $t$, but it arrived at the router before time $t-d$. 
It is not hard to see that either it is still waiting to traverse the input stage at time $t-\frac{d}{2}$, or it arrived at the middle stage by time $t-\frac{d}{2}$.
Hence the probability that the end-to-end delay is at least $d$ is at most, 
\begin{eqnarray*}
&&Pr[D_{ij}^{(1)}(t-\frac{d}{2})\ge \frac{d}{2}]+
Pr[D_{jk}^{(2)}(t)\ge \frac{d}{2}]
\\
&\le&
f_{m,\alpha}(\frac{d\alpha}{2m})+
g_{m,\beta}(\frac{d\beta}{2m}).
\end{eqnarray*}

\subsection{Characterizing the tradeoff with the number of middle elements}
In the above delay analysis we made a number of simplifying
assumptions to keep the notation manageable. For example we held
$m'=m$ and $\rbar=1$. We now demonstrate that we can use the above
formulas to handle the case of arbitrary $m'$ and $\rbar$. 
This in turn allows to characterize the tradeoff between energy consumption and end-to-end delay. 
The main
idea is to scale time so that the arrival rate at the inputs is scaled
to $1$. We also adjust the link rates on the two stages of the mesh.
In particular, if we wish to analyze a system with a given $m$, $m'$,
$\rbar$, $\alpha$ and $\beta$, we define a new system characterized by
$\hat{m}$, $\hat{m'}$, $\hat{\rbar}$, $\hat{\alpha}$ and $\hat{\beta}$
in which we set,
\begin{eqnarray*}
\hat{m'}=\hat{m}=m ~~~~~~ \hat{\rbar}=1 ~~~~~~~
\hat{\alpha}={{\alpha m'}\over {m\rbar}} ~~~~~~~
\hat{\beta}={{\beta m'}\over{m\rbar}}
\end{eqnarray*}
Then it is not hard to see that if we scale time by a factor $\rbar$,
the new system has exactly the same behavior as the old one. However,
we are now working in a system with $\hat{m'}=\hat{m}$ and
$\rbar=1$. Hence we can apply the analysis that we have already
derived. 

Our main result is thus,
\begin{theorem}
If we run the router with $m'$ middle stage elements then it uses energy $w(m')$ and the probability that the end-to-end delay is at least $d$ is bounded by,
$$
f_{m',{{\alpha m'}\over{m\rbar}}}(\frac{d\alpha}{2m\rbar})+
g_{m',{{\beta m'}\over{m \rbar}}}(\frac{d\beta}{2m\rbar}).
$$
\end{theorem}

\subsection{Dealing with dependence}
\label{s:delay-exact}
In our analysis of the middle stage we implicitly made the
simplifying assumption that $\sum_{i'} A_{i'jk}^{(1)}(s-d,t)$ is
independent from $D_{ij}(s)$. However, this is not strictly true since
the arrivals for the path $ijk$ will affect both $\sum_{i'}
A_{i'jk}^{(1)}(s-d,t)$ as well as $D_{ij}(s)$.  Since the $ijk$ flow
represents only a $1/m$ fraction of the traffic that contributes to
$D_{ij}(s)$, this will typically have a negligible effect on the
eventual results. However, if we required a true upper bound on the
probability distribution of $Q_{jk}^{(2)}$ we must use the following
adaptation of the formula, which for each $\hat{\imath}$, conditions
the event $D_{\hat{\imath}j}^{(1)}(s)\ge d$ on whether or not
$A_{\hat{\imath}jk}^{(1)}(s-d,t)$ is greater than
$\frac{5(t-s)\beta}{mn}$. More formally,

\begin{eqnarray*}
&& Pr[Q_{jk}^{(2)}(t)\ge q]\\
&\le&
\sum_d\sum_{s\le t} \left( \sum_{\hat{\imath}}Pr[A_{\hat{\imath}jk}^{(1)}(s-d,t)\ge \frac{5(t-s)\beta}{mn}] + \right.\\
&& \sum_{\hat{\imath}}
Pr[\sum_{i\neq \hat{\imath}} A_{ijk}^{(1)}(s-d,t)\ge 
 \frac{(t-s)\beta}{m}+q-\frac{5(t-s)\beta}{mn}]\times\\
&&\left.Pr[D_{\hat{\imath}j}^{(1)}((s))\ge d|A_{\hat{\imath}jk}^{(1)}(s-d,t)\le \frac{5(t-s)\beta}{mn}]\right) 
\end{eqnarray*}

%

\subsection{Queues at output}

We conclude this section by explaining how the analysis can be extended if we also wish to bound the delay on the output link from the router. 
Let $Q_k^{(3)}(t)$ be the queue at the head of the
output link and recall that this link has speed $1$. To bound $Q_{k}^{(3)}(t)$ at
time $t$, let $s\le t$ be the last time that queue $Q_{k}^{(3)}$ was
empty. For $Q_k^{(3)}$ to be larger than $q$, there must be at least
$t-s+q$ distinct packets in $Q_k^{(3)}$ during the time period
$[s,t]$. Further, let $s-d$ be the earliest time one of these packets
arrived at an input. Suppose that the path of this packet is
$ijk$. Then the packet must either be waiting in the queue
$Q_{ij}^{(1)}$ at time $s-\frac{d}{2}$ or it must be waiting in the
queue $Q_{jk}^{(2)}$ at time $s-\frac{d}{2}$. In the former case the
packet must have experienced delay of $\frac{d}{2}$ in $Q_{ij}^{(1)}$
at time $s-\frac{d}{2}$ and in the latter case it must have
experienced delay of $\frac{d}{2}$ in $Q_{jk}^{(2)}$ at time $s$.
Therefore,
\begin{eqnarray*}
&& Pr[Q_k^{(3)}(t)\ge q] \\
&\le & \sum_d \sum_{s\le t} Pr
\left[
\sum_{ij} A_{ijk}(s-d,t)\ge t-s+q
\right]\cdot \\
&& 
\left(
Pr[\exists ij, D_{ij}^{(1)}(s-\frac{d}{2})\ge \frac{d}{2}]+ \right. 
\left. Pr[\exists j, D_{jk}^{(2)}(s)\ge \frac{d}{2}]
\right)\\
&\le& \sum_d nmf_{m,\alpha}(\frac{d\alpha}{2m})\cdot mg_{m,\beta}(\frac{d\beta}{2m})\cdot\frac{1}{\ve}(d-q+\sigma_k).
\end{eqnarray*}
In addition, as before, we can convert this bound to a bound on delay for the third stage.  
Let
$D_{k}^{(3)}$ be the maximum delay that some packet has experienced
at time $t$ in the queue $Q_{k}^{(3)}(t)$.  For this delay to be more
than $d$ at time $t$, some packet must be in the queue at time
$t-d$. Since link $k$ operates in a FIFO manner, the queue at time
$t-d$ must be at least $d(1-\ve)$.
Hence $Pr[D_{k}^{(3)}(t)\ge d]\le Pr[Q_{k}^{(3)}(t-d)\ge d(1-\ve)]\le 
\sum_{d'} nmf_{m,\alpha}(\frac{d'\alpha}{2m})\cdot mg_{m,\beta}(\frac{d'\beta}{2m})\cdot\frac{1}{\ve}(d'-d(1-\ve)+\sigma_k)$.

\section{Numerical Results}
\label{s:numerical}

In this section we show numerical examples of the queue bounds.  In
particular for these calculations we use formulas that are similar to
those derived in Section~\ref{s:delay-approx} but we use Chernoff
bounds of the form (\ref{eq:cb1}) rather than (\ref{eq:cb2}) since the
former are tighter.

Figure~\ref{f:queue} plots the logarithm of the probability of a
middle-stage queue $jk$ exceeding $q$ packets against the queue size
$q$.  
For this instance,
the router is $20\times 80\times 20$.  The traffic is fully loaded for
each of the inputs and outputs.  We vary the speedup $\alpha=\beta$
in the range of $2, 3, 4$ and $5$. As we can see from the figure the
logarithm of the probability decreases linearly with the queue size,
which means the probability decreases exponentially with the queue. As
expected, we can also see that with increasing speedup, the
probability of the queue size exceeding $q$ drops.

\begin{figure}
\begin{center}
\includegraphics[width=3.5in]{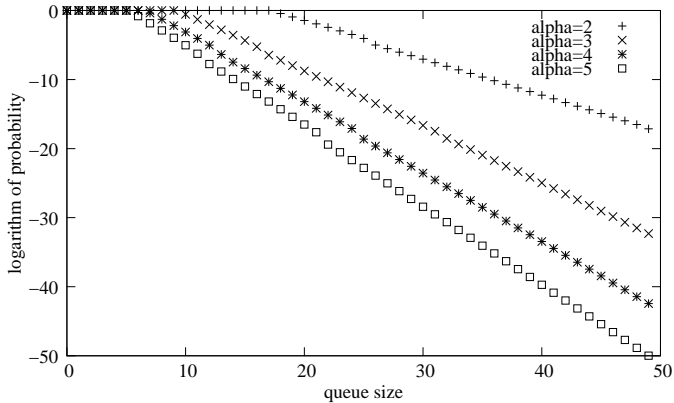}
\end{center}
\vspace{-0.6cm}
\caption{Log of probability against queue size, for $n=20$, $m=80$ and
fully loaded traffic. From top to bottom, the curves correspond to
increasing speedup from $\alpha=\beta=2$, $3, 4$ and $5$.  }
\label{f:queue}
\end{figure}

Figure~\ref{f:queue-power} demonstrates the tradeoff between the queue
size and the number of active middle-stage nodes. For this instance,
the router is again $20\times 80\times 20$. The link rate of the
interconnect is set to $1/20$.  The traffic is fully loaded for each
of the inputs and outputs.  If we keep all $m=80$ nodes in the middle
stage active, we can see from the bottom curve of
Figure~\ref{f:queue-power} that the queue size is the
smallest. However, this option is also the most energy consuming as
all 80 middle-stage nodes are kept active.  For the other extreme, we
can activate 40 middle-stage nodes, which is the most energy
efficient.  However, we can see from the bottom curve of
Figure~\ref{f:queue-power} that the queue size is the smallest.  The
curves in between correspond to the intermediate cases in which the
number of active middle-stage nodes are 50, 60 and 70.  

\begin{figure}
\begin{center}
\includegraphics[width=3.5in]{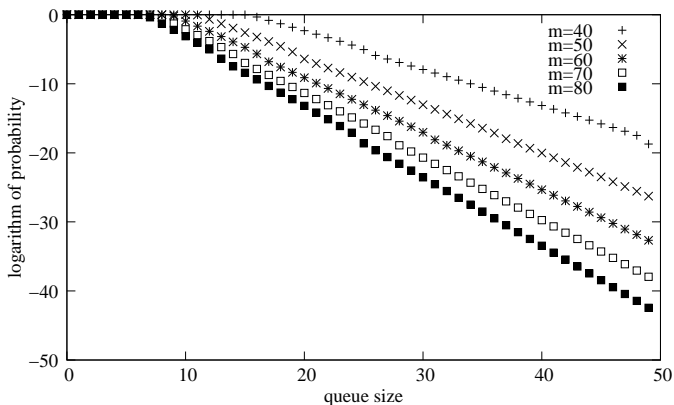}
\end{center}
\vspace{-0.6cm}
\caption
{Log of probability against queue size, for same amount of total
traffic but varying number of active middle-stage nodes. From
top to bottom, the curves correspond to increasing number of
active middle-stage nodes from $m=40, 50, 60, 70, 80$. The number
of inputs and outputs is $n=20$ and interconnect link rate is $1/20$.}
\label{f:queue-power}
\end{figure}

\section{Conclusion}

In this paper we revisit the Load-Balanced Router architecture,
motivated by its potential of delivering energy proportionality for
routers.  We offer a detailed analysis on the queue lengths and packet
delays under a simple random routing algorithm which is robust against
all admissible traffic. This allows us to observe a trade off between
performances such as queue size against energy consumption. 

Our paper does not focus on algorithms that optimize the number of
active middle-stage nodes.  We give a very simple argument for which
the size of the active middle stage is proportional to the {\em
maximum} traffic load over all inputs and outputs.  It is an intriguing open
question to see how one could make sure that the size of the active
middle stage is proportional to the traffic average over all input, not
to the maximum.

\bibliographystyle{abbrv}
\bibliography{/home/ylz/GREEN/refs,/home/ylz/mainrefs}

\cm{
\section{Appendix}

We make use of the following equations in deriving the probability bounds.

\begin{eqnarray*}
\sum_{i=a}^{\infty} z^i&=& \frac{z^a}{1-z} \\ 
\sum_{i=a}^{\infty}
iz^{i-1} &=& \frac{az^{a-1}-(a-1)z^a}{(1-z)^2}\\
\sum_{i=a}^{\infty} i(i-1)z^{i-2} &=& 
\end{eqnarray*}
\begin{eqnarray*}
\hspace{1cm}\frac{(a^2-a)z^{a-2}+(-2a^2+4az)z^{a-1} + (a^2-3a+2)z^a}{(1-z)^3}
\end{eqnarray*}
}

\cm{
\begin{eqnarray*}
\sum_{i=a}^{\infty} z^i&=& \frac{z^a}{1-z},\\

\sum_{i=a}^{\infty}
iz^{i-1} &=& \frac{d}{dz}\sum_{i=a}^{\infty} z^i\\ &=&
\frac{d}{dz}\frac{z^a}{1-z}\\ &=& \frac{az^{a-1}-(a-1)z^a}{(1-z)^2},\\

\sum_{i=a}^{\infty} i(i-1)z^{i-2} &=&
\frac{d^2}{dz^2}\sum_{i=a}^{\infty} z^i\\ &=&
\frac{d^2}{dz^2}\frac{z^a}{1-z}\\ &=&
\frac{d}{dz}\frac{az^{a-1}-(a-1)z^a}{(1-z)^2}\\ &=&
\frac{(1-z)^2(a(a-1)z^{a-2}-a(a-1)z^(a-1))+2(1-z)(az^{a-1}-(a-1)z^a)}
{(1-z)^4}\\
&=&\frac{a^2z^{a-2}-az^{a-2}-a^2z^{a-1}+az^{a-1}-a^2z^{a-1}+az^{a-1}+
a^2z^a-az^a+2az^{a-1}-2az^a+2z^a}{(1-z)^3}\\
&=& \frac{(a^2-a)z^{a-2}+(-2a^2+4az)z^{a-1}+(a^2-3a+2)z^a}{(1-z)^3}.
\end{eqnarray*}
}
\end{document}